\documentclass[
]{cc}
\usepackage[colorlinks=true,citecolor=blue, linkcolor=blue
]{hyperref}

\def\ba{\ensuremath{\mathbf{a}}}

\def\th{{\tilde h}}
\def\tT{{\tilde T}}
\def\tlambda{{\tilde\lambda}}

\def\bb{{\mathbf{b}}}
\def\cc{{\mathbf{c}}}
\def\xx{{\mathbf{x}}}
\def\zz{{\mathbf{z}}}
\def\PP{{\mathbf{P}}}
\def\supp{{\operatorname{supp}}}

\newcommand{\smallbmatrix}[1]{{
\left[\begin{smallmatrix} #1 \end{smallmatrix}\right]}}

\title{On semiring complexity of Schur polynomials}

\contact{fomin@umich.edu}
\submitted{20 February 2017}

\author{Sergey Fomin\\
Department of Mathematics\\
University of Michigan\\
Ann Arbor, MI 48109, USA\\
\email{fomin@umich.edu}\\
\and
Dima Grigoriev\\
CNRS, Math\'ematiques\\
Universit\'e de Lille\\
Villeneuve d'Ascq, 59655, France\\
\email{Dmitry.Grigoryev@math.univ-lille1.fr}\\
\and
Dorian Nogneng\\
LIX, 
\'Ecole Polytechnique\\
91128 Palaiseau Cedex, France\\
\email{dorian.nogneng@lix.polytechnique.fr}\\
\and
\'Eric Schost\\
Cheriton School of Computer Science\\
University of Waterloo\\
Waterloo, ON, Canada N2L 3G1\\
\email{eschost@uwaterloo.ca}
}

\authorhead{Fomin, Grigoriev, Nogneng \& Schost}

\begin{subject}
\end{subject}

\begin{abstract}
Semiring complexity is the version of arithmetic
\hbox{circuit} complexity
that allows only 
two operations: addition and 
multiplication.
We show that semiring complexity of a Schur polynomial
$s_\lambda(x_1,\dots,x_k)$ labeled by a partition
$\lambda=(\lambda_1\ge\lambda_2\ge\cdots)$
is bounded by~$O(\log(\lambda_1))$ provided the 
number of variables $k$ is fixed. 
\end{abstract}


\begin{document}

\section{Introduction and main results}

Let $f(x_1,\dots,x_k)
$ be a polynomial with nonnegative integer
coefficients. 
As such, $f$ can be computed using addition and
multiplication only---without subtraction or division.
To be more precise, one can build an \emph{arithmetic circuit} wherein
\begin{itemize}
\item
each gate performs an operation of addition or multiplication; 
\item
the inputs are $x_1,\dots,x_k$, possibly along with some positive
integer scalars;  
\item
the sole output is $f(x_1,\dots,x_k)$. 
\end{itemize}
The \emph{semiring complexity} (or \emph{$\{+,\times\}$-complexity}) of~$f$ is the smallest size of (i.e.,
the smallest number of
gates in) such an arithmetic circuit. 
This notion is illustrated in \ref{fig:arithmetic-circuit-h5}. 
For additional details, see \citet[Section~2]{Fomin-Grigoriev-Koshevoy} and
references therein. 

\begin{figure}[ht]
\begin{center}
\setlength{\unitlength}{1.5pt} 
\begin{picture}(60,110)(0,0) 
\thinlines 

\put(0,-2){\makebox(0,0){$x_1$}}
\put(20,-2){\makebox(0,0){$x_2$}}


\put(0,20){\circle{10}}
\put(0,20){\makebox(0,0){$+$}}

\put(20,20){\circle{10}}
\put(20,20){\makebox(0,0){$\times$}}

\put(40,20){\circle{10}}
\put(40,20){\makebox(0,0){$\times$}}

\put(40,40){\circle{10}}
\put(40,40){\makebox(0,0){$+$}}

\put(60,40){\circle{10}}
\put(60,40){\makebox(0,0){$\times$}}

\put(20,60){\circle{10}}
\put(20,60){\makebox(0,0){$\times$}}

\put(20,80){\circle{10}}
\put(20,80){\makebox(0,0){$+$}}

\put(0,100){\circle{10}}
\put(0,100){\makebox(0,0){$\times$}}

\put(0,3){\vector(0,1){10}}

\put(3.5,1.5){\vector(1,1){13}}
\put(2,3){\vector(1,1){13}}

\put(23.5,1.5){\vector(1,1){13}}
\put(22,3){\vector(1,1){13}}

\put(45.5,23.5){\vector(1,1){11}}
\put(44,25){\vector(1,1){11}}

\put(17,2){\vector(-1,1){12}}

\put(25,25){\vector(1,1){10}}
\put(35,45){\vector(-1,1){10}}
\put(15,85){\vector(-1,1){10}}
\put(55,45){\vector(-1,1){30}}

\put(20,27){\vector(0,1){26}}
\put(0,27){\vector(0,1){66}}
\put(0,107){\vector(0,1){7}}
\put(40,26.5){\vector(0,1){7}}
\put(20,66.5){\vector(0,1){7}}

\end{picture}
\end{center}
\caption{\label{fig:arithmetic-circuit-h5}
The smallest $\{+,\times\}$-circuit 
computing the polynomial 
$f(x_1,x_2)=h_5(x_1,x_2)
=x_1^5+x_1^4x_2+x_1^3x_2^2+x_1^2x_2^3+x_1x_2^4+x_2^5
$.
This circuit utilizes the formula
$h_5(x_1,x_2)=(x_1+x_2)(x_1^2(x_1^2+x_2^2)+x_2^4)$.}
\end{figure}

This paper is devoted to the problem of determining semiring
complexity of symmetric polynomials. 
More specifically, we focus our attention on Schur functions,
an important class of symmetric polynomials which play prominent roles
in 
several branches of mathematics; see, e.g.,
\citet[Chapter~I]{Macdonald} and \citet[Chapter~7]{EC2}. 


Let $\lambda=(\lambda_1\ge \lambda_2\ge\cdots\ge 0)$ be an integer
partition. 
The \emph{Schur function} (or \emph{Schur polynomial}))
$s_\lambda(x_1,\dots,x_k)$
is a symmetric polynomial of degree $|\lambda|=\sum_i \lambda_i$ 
in the variables $x_1,\dots,x_k$ 
which can be defined in many different ways. 
One remarkable feature of Schur polynomials that makes them an
exciting object of study in algebraic complexity theory is that the
classical formulas defining them fall into two categories. 
On the one hand, there are 
determinantal expressions (e.g., the Jacobi-Trudi formula or the
bialternant formula) which provide efficient ways to compute Schur functions
in an unrestricted setting, i.e., when all arithmetic operations are
allowed.   
On the other hand, Schur functions 
are generating functions for semistandard Young tableaux.  
This description represents them as polynomials with manifestly positive
coefficients; so they can be computed using addition and multiplication only.
We note however that the na\"ive approach based on these monomial
expansions yields algorithms whose (semiring) complexity is very high---and
indeed very far from the optimum. 




Our main result is the following. 
(We use the notation $\lambda'=(\lambda_1'\ge\lambda_2'\ge\cdots)$ for the partition conjugate to~$\lambda$.) 

\begin{theorem}
\label{th:main}
The semiring complexity of a Schur polynomial
$s_\lambda(x_1,\dots,x_k)$ labeled \linebreak[3] by 
partition~$\lambda\!=\!(\lambda_1\!\ge\!\cdots\!\ge\!\lambda_\ell)$
is at most $O(\log(\lambda_1) k^5 2^{k\ell}\ell^d)$
where \hbox{$d\!=\!\displaystyle\max_j \lambda'_j(k\!-\!\lambda'_j)$}. 
\end{theorem}

Since $\ell\le k$ (or else $s_\lambda(x_1,\dots,x_k)=0$)
and $d\le \frac{k^2}{4}$, we obtain: 
 
\begin{corollary}
\label{cor:main}
The semiring complexity of a Schur polynomial
$s_\lambda(x_1,\dots,x_k)$  is bounded from above 
by $k^{k^2 (\frac14+o(1))} O(\log(\lambda_1))$.
If the number of variables $k$ is fixed, this complexity is~$O(\log(\lambda_1))$. 
\end{corollary}


\begin{remark}
The problem of designing efficient algorithms employing addition and multiplication 
arises naturally in the context of numerical computation,
as these algorithms have valuable stability properties. 
Motivated by such considerations, \citet{Demmel-Koev} developed 
$\{+,\times\}$-algorithms for computing Schur polynomials 
using a dynamic programming approach.
In the notation of \ref{th:main}, Proposition~5.3 in \emph{loc.\ cit.}\ 
asserts that the semiring complexity of 
$s_\lambda(x_1,\dots,x_k)$ 
is bounded from above by $O(e^{5.2\sqrt{|\lambda|}}\ell k)$. 
When $k$ is fixed, and the shape~$\lambda$ grows, this bound is much larger than 
the one in \ref{cor:main}.
On the other hand, in the regime where $\lambda$ is fixed and the number of variables~$k$
grows, the complexity of the Demmel-Koev algorithm is linear in~$k$ whereas 
the bound in \ref{th:main} is exponential in~$k$. 
It would be interesting to find a common generalization of these results. 
\end{remark}

We prove \ref{th:main} in two stages.
At the first stage (see \ref{sec:h_n}), 
we treat a special case where partition~$\lambda$ has only one
(nonzero) part. 
More explicitly, we obtain the following result. 

Recall that the \emph{complete homogeneous symmetric polynomial}
\begin{equation*}
\label{eq:h_n}
h_n(x_1,\dots,x_k)=\sum_{1\le i_1\le\cdots \le i_n\le k} x_{i_1}\cdots x_{i_n} 
\end{equation*}
is the sum of all monomials of degree~$n$ 
in the variables $x_1,\dots,x_k$.  
See an example in \ref{fig:arithmetic-circuit-h5}. 

\begin{theorem}
\label{th:h_n}
The semiring complexity of a complete homogeneous symmetric polynomial 
$h_n(x_1,\dots,x_k)$ is $O(k^2\log(n))$. 
\end{theorem}

Our proof of \ref{th:main}, presented in \ref{sec:proof-main}, 
relies on three main ingredients: 
\begin{itemize}
\item
\ref{th:h_n}; 
\vspace{-.05in}
\item
a formula expressing a multichain-generating function of a \emph{shellable poset}
in terms of complete homogeneous polynomials, see \ref{sec:shelling}; and
\vspace{-.05in}
\item
a representation of a Schur polynomial as a multichain generating function, 
or more precisely an iterated sum thereof, 
see \ref{sec:schur-via-multichains}. 
\end{itemize}

\section{Related problems}

The general problem of determining the semiring complexity of a Schur
polynomial is open.
In particular, the following tantalizing problem remains out of reach. 

\begin{problem}[{\citet[Problem~3.2]{Fomin-Grigoriev-Koshevoy}}]
\label{problem:+*-schur}
Is the semiring complexity of $s_\lambda(x_1,\dots,x_k)$ bounded by a polynomial
in~$k$ and~$|\lambda|$?
\end{problem}

\begin{remark}
A general method for obtaining lower bounds on semiring complexity was
suggested by \citet{Schnorr}. 
Schnorr's bound only depends on the \emph{support} of a polynomial,
i.e., on the set of monomials that contribute with a positive
coefficient.
Schnorr's argument was further refined 
by \citet{Shamir-Snir};
powerful applications were given by \citet{Jerrum-Snir}.
As mentioned in \citet[Remark~3.3]{Fomin-Grigoriev-Koshevoy}, 
Schnorr-type lower bounds are useless in the case
of Schur functions since computing a Schur function is difficult not because of
its support  but because of the complexity of its coefficients (the
Kostka numbers).
The problem of computing an individual Kostka number is known to be
\#\textbf{P}-complete \citep{Narayanan} whereas the support of a Schur
function is very easy to determine.
\end{remark}


\begin{remark}
\citet{Fomin-Grigoriev-Koshevoy} investigated the notion of
semi\-ring complexity alongside other similar
computational models involving restricted sets of arithmetic
operations. 
In brief, the results obtained in \emph{loc.\ cit.}, together with
\citet{Jerrum-Snir} and \citet{Valiant}, 
demonstrate that adjoining subtraction and/or division to the two-element set
$\{+,\times\}$ of allowed arithmetic
operations can, in some cases, dramatically decrease computational complexity. 
(By contrast, removing division from $\{+,-,\times,\div\}$ comes at merely polynomial cost,
as shown by \citet{Strassen}.)
We refer the reader to~\citet{Fomin-Grigoriev-Koshevoy} for the discussion of these issues.  
\end{remark}

\begin{remark}
In the unrestricted model, one can compute a Schur polynomial $s_\lambda(x_1,\dots,x_k)$ 
in time polynomial in~$k$ and $\log(\lambda_1)$,
via the bialternant formula \citep[Section 7.15]{EC2},
and using repeated squaring to compute the powers of variables appearing in the
relevant determinants. 
\end{remark}

One important complexity model studied in~\citet{Fomin-Grigoriev-Koshevoy} is 
\emph{subtraction-free complexity}, which allows 
the operations of addition, multiplication, and division. 
It turns out that subtraction-free complexity of a Schur function is indeed
polynomial:

\begin{theorem}[{\citep[Section 6]{Koev-2007}, \citep[Section 4]{CDEKK},
  \citep[Theorem~3.1]{Fomin-Grigoriev-Koshevoy}}]
\label{th:schur}
Subtraction-free complexity of a Schur polynomial
$s_\lambda(x_1,\dots,x_k)$ is at most $O(n^3)$ where $n=k+\lambda_1$.
\end{theorem}

The algorithms presented in \emph{loc.\ cit.}\ 
utilize division in essential ways,
so they do not bring us any closer to the resolution of \ref{problem:+*-schur}. 

Since subtraction-free
complexity is bounded from above by semiring complexity, 
\ref{th:main} implies that the  subtraction-free
complexity of a particular Schur polynomial
$s_\lambda(x_1,\dots,x_k)$
can be much smaller (for small~$k$) than the upper bound of \ref{th:schur}.

\begin{problem}
Find a natural upper bound on subtraction-free
  complexity of a Schur polynomial that simultaneously strengthens
  \ref{th:main} and \ref{th:schur}.
\end{problem}

\begin{remark}
\citet{Grigoriev-Koshevoy} gave an exponential lower bound on 
the $\{+,\times\}$-complexity of a monomial symmetric function.
\end{remark}

\pagebreak[3]

\section{Semiring complexity of complete homogeneous polynomials} 
\label{sec:h_n}

In this section, we prove \ref{th:h_n}. 
We fix~$k$, and use the notation
\begin{align*}
h_m&=h_m(x_1,\dots,x_k)=\sum_{1\le i_1\le\cdots \le i_m\le k} x_{i_1}\cdots x_{i_m}\,,\\
\th_m&=h_m(x_1^2,\dots,x_k^2), \\
e_m&=e_m(x_1,\dots,x_k)=\sum_{1\le i_1<\cdots < i_m\le k} x_{i_1}\cdots x_{i_m}\,.  
\end{align*}
\begin{lemma}
\label{lem:h_n-recursive}
One can compute the polynomials
$\,h_{n-k+1},\dots,h_{n}\,$
from $\th_{\lfloor \frac{n}{2}\rfloor -k+1},\dots,\th_{\lfloor \frac{n}{2}\rfloor}$
and $e_1,\dots,e_k$,  using $O(k^2)$ additions and multiplications. 
\end{lemma}

\begin{proof}
The key algebraic observation is that 
\[
\sum_{m \ge 0} h_m \,t^m \!=\! \prod_{i=1}^k 
\frac{1}{1\!-\!x_i t}
= \prod_{i=1}^k 
(1+x_i t) \prod_{i=1}^k 
\frac{1}{1\!-\!x_i^2 t^2} 
= \sum_{a=0}^k e_a t^a  \sum_{b \ge 0} \th_b t^{2b}
\]
and consequently
\begin{equation}
\label{eq:h_m-recursion}
h_m=\sum_{m-k\le 2b\le m} e_{m-2b} \, \th_b\,.
\end{equation}
For $n\!-\!k\!+\!1\le m\le n$, the indices $b$ 
appearing on the right-hand side of~\ref{eq:h_m-recursion}
satisfy $b\le \lfloor \frac{m}{2}\rfloor\le \lfloor \frac{n}{2}\rfloor$
and $b\ge \lceil \frac{m-k}{2} \rceil\ge \lceil \frac{n-2k+1}{2} \rceil
=\lfloor\frac{n}{2}\rfloor-k+1$.
Thus we can use \ref{eq:h_m-recursion} to compute these~$h_m$;
this takes $O(k)$ operations for each of the $k$ values of~$m$. 
\end{proof}

\begin{lemma}
\label{lem:e-semiring-complexity}
One can compute 
$e_1,\dots,e_k$  using $O(k^2)$ additions and multiplications. 
\end{lemma}

\begin{proof}
The requisite algorithm is obtained by iterating the Pascal-type recurrence
\[
e_m(x_1,\dots,x_j)=x_j e_{m-1}(x_1,\dots,x_{j-1})+e_m(x_1,\dots,x_{j-1}). \qed
\]
\end{proof}

We note that in the unrestricted model, the complexity of computing $e_1,\dots,e_k$
is of the order $k\log(k)$, see~\citet{Strassen-E}. 

\begin{proof}[\ref{th:h_n}]
Let $T(n)$ denote the semiring complexity of computing \linebreak[3]
$h_{n-k+1},\dots,h_{n}$. 
\ref{lem:h_n-recursive} and \ref{lem:e-semiring-complexity}
imply that $T(n) \le T(\lfloor \frac{n}{2} \rfloor) +O(k^2)$.
(Squaring the variables~$x_1,\dots,x_k$, which is needed to compute the 
$\th_b$'s, takes linear time.)
We conclude that $T(n) = O(k^2 \log(n))$, as desired.
\end{proof}


\section{Linear orderings of maximal chains in partially ordered sets} 
\label{sec:shelling}

\begin{definition}[\emph{Poset,  chain, proper ordering}]
\label{def:chain}
Let $\PP$ be a finite graded partially ordered set (\emph{poset}) with a unique
minimal element~$\hat 0$ and a unique maximal element~$\hat 1$. 
A~linearly ordered subset of~$\PP$ 
is called a \emph{chain}. 
We denote by $\operatorname{MaxChains}(\PP)$ the~set of  all maximal (by inclusion) chains in~$\PP$.
Under the above assumptions, all chains in $\operatorname{MaxChains}(\PP)$ have the
same cardinality~$m$. 

Let us fix a linear ordering 
on $\operatorname{MaxChains}(\PP)$,
and write $Q'\prec Q$ to denote that $Q'$ (strictly) precedes~$Q$ in this order. 
For $Q\in\operatorname{MaxChains}(\PP)$, we denote
\begin{equation}
\label{eq:Q*}
Q^*\stackrel{\rm def}{=}
\{\cc\in Q \mid Q-\{\cc\}\subset Q' \text{~for~some~} Q'\prec Q\}. 
\end{equation}
Thus $Q^*$ consists of those elements of a maximal chain~$Q$ which can be replaced by
another element so that the resulting maximal chain precedes~$Q$. 
We call a linear ordering of $\operatorname{MaxChains}(\PP)$
\emph{proper} if for any~$Q\in\operatorname{MaxChains}(\PP)$, 
none of the chains preceding~$Q$ contains~$Q^*$: 
\begin{equation}
\label{eq:proper-ordering}
Q'\prec Q \Longrightarrow Q'\not\supset Q^*. 
\end{equation}
\end{definition}

\begin{remark}
In algebraic/geometric combinatorics, 
the notions introduced in \ref{def:chain} 
are traditionally described in the language of simplicial complexes and their shellings;
see, e.g., \citet{wachs} for an introduction to this subject. 
In this paper, we try to avoid this terminology in order to keep the exposition self-contained. 
The brief comments below are intended for the readers interested in the broader 
combinatorial context, and will not be relied upon in the sequel. 

The \emph{order complex} 
of~$\PP$ is the simplicial complex on the
ground set~$\PP$ whose simplices are the chains in~$\PP$. 
The maximal simplices of the order complex are the maximal chains. 
A~linear ordering of $\operatorname{MaxChains}(\PP)$
is called a \emph{shelling} (of the order complex) if
for any $Q\in\operatorname{MaxChains}(\PP)$, the subcomplex of the order complex 
formed by the simplices~$Q'$ with $Q'\prec Q$ 
(or more precisely the geometric realization of this subcomplex) 
intersects (the geometric realization of) 
the maximal simplex~$Q$ at a union of codimension~1 faces of~$Q$. 
It~is well known---and not hard to see---that any shelling order is proper, 
in the sense of \ref{def:chain}. 
More concretely, the subchain $Q^*\subset Q$ defined via~\ref{eq:Q*} can be seen to coincide with 
the complement (inside~$Q$)
of the intersection of the aforementioned codimension~1 faces. 
Put differently, $Q^*$ is the unique smallest face of $Q$ not contained in 
the subcomplex $\bigcup_{Q'\prec Q} Q'$. 
\end{remark}

Our use of the notion of a proper ordering of maximal chains
will rely on the following key lemma. 

\pagebreak[3]

\begin{lemma}
\label{lem:key-lemma}
Let $\PP$ be a poset with a proper linear ordering on $\operatorname{MaxChains}(\PP)$,
as in \ref{def:chain}. 
For a chain~$C$ and a maximal chain~$Q$, 
the following are equivalent:
\begin{itemize}
\item[{\rm (i)}]
$Q$ is the smallest maximal chain  containing~$C$
(with respect to the linear ordering on $\operatorname{MaxChains}(\PP)$); 
\vspace{-.05in}
\item[{\rm (ii)}]
$Q^*\subset C\subset Q$ (recall that $Q^*$ is defined by~\ref{eq:Q*}). 
\end{itemize}
\end{lemma}

\begin{proof}
First assume that (i) holds.
Let $\cc\in Q^*$. If $\cc\notin C$, then $C$ is contained in some maximal chain
$Q'\prec Q$ (see~\ref{eq:Q*}), contradicting~(i). 

Going in the opposite direction, assume that $Q^*\subset C\subset Q$. 
Suppose there exists a maximal chain $Q'\prec Q$ containing~$C$.
Then $Q'\supset Q^*$, contradicting~\ref{eq:proper-ordering}. 
\end{proof}

\begin{definition}[\emph{Multichain, support}]
A ``weakly increasing'' sequence 
\[
M=\{p_1\le\cdots \le p_m\}\subset\PP
\]
is called a \emph{multichain} of size~$m$; we write $|M|=m$. 
The elements of $\PP$ which appear in $M$ (with nonzero
multiplicity) form the \emph{support} of~$M$, denoted
by~$\supp(M)$. 
The support of a multichain is a chain. 
\end{definition}

Let us associate a formal variable $z_\cc$ with each element \hbox{$\cc\in \PP$}. 
For a multiset $M$ of elements in~$\PP$, we denote by $\zz^M$ the corresponding
monomial: $\zz^M=\prod_{\cc\in M} z_\cc$. 

\begin{lemma}
\label{lem:gf-shelling}
Let $\PP$ be a poset endowed with a proper linear ordering of its maximal chains,
see \ref{def:chain}. 
(Or: assume that a shelling of the order complex of~$\PP$ is given.)
Then 
the generating function 
for the multichains of size~$m$ in~$\PP$ is given by 
\begin{equation}
\label{eq:multichains-via-h}
\sum_{\substack{\text{\rm multichain $M$}\\[.03in] |M|=m}} \zz^M
= \sum_{Q\in \operatorname{MaxChains}(\PP)}
\ \zz^{Q^*} h_{m-|Q^*|}((z_\cc)_{\cc\in Q}),  
\end{equation}
with $Q^*$ defined by~\ref{eq:Q*}. 
\end{lemma}

\begin{proof}
By \ref{lem:key-lemma}, 
the set of chains in~$\PP$ splits into the disjoint union of 
(poset-theoretic) intervals of the form~$[Q^*,Q]$. 
Categorizing the multichains~$M$ by their support,
and applying this observation to $C=\supp(M)$, we obtain the identity
\begin{equation*}
\sum_{\substack{\text{\rm multichain $M$}\\[.03in] |M|=m}} \zz^M
= \sum_{Q\in \operatorname{MaxChains}(\PP)}
\ \sum_{\substack{Q^*\subset\operatorname{supp}(M)
\subset Q\\
|M|=m}} \zz^M ,
\end{equation*}
which readily implies~\ref{eq:multichains-via-h}.  
\end{proof}

In \ref{sec:schur-via-multichains},
we will relate Schur polynomials to a special case of the
above construction
involving a class of (shellable) posets $\PP_{h,k}$ described
in \ref{def:P_{hk}} below. 
These posets have been extensively studied in algebraic
combinatorics, due to the role they play in representation theory and
the classical Schubert Calculus.  
In particular, $\PP_{h,k}$ describes the attachment of Schubert cells
in the Grassmann manifold~$\operatorname{Gr}(h,k)$.  

\pagebreak[3]

\begin{definition}[\emph{Posets $\PP_{h,k}$}]
\label{def:P_{hk}}
Let $h$ and $k$ be positive integers, with $h\le k$. 
We denote by $\PP_{h,k}$ the poset whose elements are column vectors
(or simply \emph{columns}) $\cc$~of \emph{height}~$h$ 
whose entries lie in the set $\{1,\dots,k\}$ and strictly increase
downwards: 
\begin{equation}
\label{eq:columns}
\cc=\smallbmatrix{c_1\\[-.02in] \vdots\\[.05in] c_h}\in\mathbb{Z}^h,\quad
1\le c_1<\cdots<c_h\le k;  
\end{equation}
by definition, 
$\smallbmatrix{c_1\\[-.02in] \vdots\\[.05in] c_h}
\le 
\smallbmatrix{c_1'\\[-.02in] \vdots\\[.05in] c_h'}$
if and only if 
$\left\{
\begin{smallmatrix}c_1\le c_1'\\[-.02in] \vdots\\[.05in] c_h\le c_h'
\end{smallmatrix}
\right.$. 
\end{definition}

Let us make a few simple but useful observations. 

\begin{lemma}
{\ }
\begin{enumerate}
\item
The cardinality of $\PP_{h,k}$ is~$\binom{k}{h}$. 
\item
The columns
$\hat 0=\smallbmatrix{1\\ \vdots\\[.05in] h}$ and
$\hat 1=\smallbmatrix{k-h+1\\ \vdots\\[.05in] k}$
are the unique minimal and maximal elements of $\PP_{h,k}$, respectively. 
\item
The poset $\PP_{h,k}$ is graded,
with the rank function given by 
\[
\operatorname{rk}(\cc)=c_1+\cdots+c_h-\frac{h(h+1)}{2}.
\]
\item
Each maximal chain in $\PP_{h,k}$ has cardinality $h(k-h)+1$. 
\end{enumerate}
\end{lemma}

\begin{remark}
\label{rem:chains-into-syt}
The poset $\PP_{h,k}$ is canonically isomorphic to the poset of
integer partitions 
(partially ordered component-wise) 
having at most $h$ parts all of which are $\le k-h$. 
The isomorphism is given by
\[
\cc
\mapsto
(c_h-h,\dots,c_1-1). 
\]
Put another way, $\PP_{h,k}$ is canonically isomorphic to the poset of 
Young diagrams fitting inside the $h\times(k-h)$ rectangle,
ordered by inclusion.
Such a Young diagram $\lambda=(\lambda_1\ge\cdots\ge\lambda_h)$
corresponds to the column~$\cc$ as in~\ref{eq:columns} described pictorially as follows. 
We assume the ``English'' convention for drawing Young diagrams, with the longest row at the top. 
Starting at the lower-left corner of the $h\times(k-h)$ box,
trace the lower-right boundary of~$\lambda$, making the total of $k$ unit steps. 
Among them, there are exactly $h$ vertical steps.
The location of the $i$th  vertical step, counting from the bottom, 
among the $k$ unit steps, 
is given by the  $i$th entry $c_i=\lambda_{h-i+1}+i$. 

Under this isomorphism, the maximal chains 
\[
Q=\{\cc_0\le\cc_1\le\cdots\le \cc_{h(k-h)}\}\in\operatorname{MaxChains}(\PP_{h,k})
\]
are interpreted as the standard Young tableaux of rectangular shape $h\times(k-h)$. 
(The reader unfamiliar with the tableau terminology is referred
to \ref{def:tableaux-schur}.) 
In~concrete terms, the column~$\cc_j$ describes (the lower-right boundary of)
the diagram formed by the entries $1,\dots,j$ of the standard tableau~$Q$. 
\end{remark}

\begin{example}
\label{ex:h=2,k=5}
Let $h=2$ and $k=5$. 
The poset $\PP_{2,5}$ consists of $\binom52=10$ elements of the form 
$\smallbmatrix{a\\ b}$, with $1\le a<b\le 5$. 
These are in bijection with partitions $\mu=(\mu_1,\mu_2)=(b-2,a-1)$
satisfying $3\ge \mu_1\ge \mu_2\ge 0$ (equivalently, Young diagrams
fitting inside the $2\times 3$ rectangle). 
There are $5$ maximal chains in $\PP_{2,5}$, corresponding to the
$5$~standard Young tableaux 
of this rectangular shape.
See \ref{fig:P25}. 
\end{example}

\begin{figure}[h]
\begin{equation*}
\begin{array}{c|c|c|c}
Q\in \operatorname{MaxChains}(\PP_{2,5})  &
\begin{array}{c}
$\!\!\!$\text{\footnotesize standard}$\!\!\!$\\ 
$\!\!$\text{\footnotesize tableau}$\!\!$
\end{array} 
&
Q^* & \text{\footnotesize descents}\\
\hline
&& \\[-.1in]
\!\smallbmatrix{1\\ 2}<\smallbmatrix{1\\ 3}<\smallbmatrix{2\\ 3}<\smallbmatrix{2\\ 4}<\smallbmatrix{3\\ 4}<\smallbmatrix{3\\ 5}<\smallbmatrix{4\\ 5}
& \smallbmatrix{135\\ 246} &\varnothing \\[.05in]
\!\smallbmatrix{1\\ 2}<\smallbmatrix{1\\ 3}<\smallbmatrix{2\\ 3}<\smallbmatrix{2\\ 4}<\smallbmatrix{2\\ 5}<\smallbmatrix{3\\ 5}<\smallbmatrix{4\\ 5}
& \smallbmatrix{134\\ 256} & \smallbmatrix{2\\ 5} & 4\\[.05in]
\!\smallbmatrix{1\\ 2}<\smallbmatrix{1\\ 3}<\smallbmatrix{1\\ 4}<\smallbmatrix{2\\ 4}<\smallbmatrix{3\\ 4}<\smallbmatrix{3\\ 5}<\smallbmatrix{4\\ 5}
& \smallbmatrix{125\\ 346} &\smallbmatrix{1\\ 4} & 2\\[.05in]
\!\smallbmatrix{1\\ 2}<\smallbmatrix{1\\ 3}<\smallbmatrix{1\\ 4}<\smallbmatrix{2\\ 4}<\smallbmatrix{2\\ 5}<\smallbmatrix{3\\ 5}<\smallbmatrix{4\\ 5}
& \smallbmatrix{124\\ 356} & \!\smallbmatrix{1\\ 4}\!<\!\smallbmatrix{2\\ 5}\! & 2, 4\\[.05in]
\!\smallbmatrix{1\\ 2}<\smallbmatrix{1\\ 3}<\smallbmatrix{1\\ 4}<\smallbmatrix{1\\ 5}<\smallbmatrix{2\\ 5}<\smallbmatrix{3\\ 5}<\smallbmatrix{4\\ 5}
& \smallbmatrix{123\\ 456} & \smallbmatrix{1\\ 5} & 3
\end{array}
\end{equation*}
\caption{Maximal chains in the poset $\PP_{2,5}$.}
\label{fig:P25}
\end{figure}

We will later need the following crude estimate. 

\begin{lemma}
\label{lem:bound-max-chains}
The number of maximal chains in $\PP_{h,k}$ does not exceed $h^{h(k-h)}$. 
\end{lemma}

\begin{proof}
At each of the $h\times(k-h)$ steps in a 
maximal chain, we add~$1$ to one of the $h$ components of a column. 
\end{proof}

\pagebreak[3]

\begin{definition}[\emph{Intervals $\PP_{h,k}[\ba,\bb]$, and lexicographic ordering 
of maximal~chains}]
For $\ba,\bb\in\PP_{h,k}$ satisfying $\ba\le\bb$, we
denote by $[\ba,\bb]=\PP_{h,k}[\ba,\bb]$ the corresponding (order-theoretic) interval:
\[
\PP_{h,k}[\ba,\bb]=\{\mathbf{c}\in\PP_{h,k}\mid \ba\le\mathbf{c}\le\bb\}.
\]
In the special case $\PP_{h,k}[\hat 0,\hat 1]=\PP_{h,k}$, 
we recover the entire poset~$\PP_{h,k}$. 

The \emph{lexicographic ordering}
on $\operatorname{MaxChains}(\PP_{h,k}[\ba,\bb])$  (denoted by the symbol~$\prec$)
is the linear order defined as follows.
Let 
\begin{align*}
Q&=\{\ba=\smallbmatrix{a_{11}\\[-.02in] \vdots\\[.05in] a_{h1}}<\cdots<\smallbmatrix{a_{1N}\\[-.02in] \vdots\\[.05in] a_{hN}}=\bb\},\\
Q'&=\{\ba=\smallbmatrix{a'_{11}\\[-.02in] \vdots\\[.05in] a'_{h1}}<\cdots<\smallbmatrix{a'_{1N}\\[-.02in] \vdots\\[.05in] a'_{hN}}=\bb\}
\end{align*}
be two maximal chains in $\PP_{h,k}[\ba,\bb]$. 
Let $j$ indicate the leftmost position where these two chains differ, 
i.e., the smallest index for which there exists $i$ with $a_{ij}\neq a'_{ij}$. 
Furthermore, let $i$ be the largest index (i.e., the lowermost location) 
for which this inequality occurs (for the minimal choice of~$j$). 
Then $Q'\prec Q$ if and only if $a'_{ij}<a_{ij}$. 
\end{definition}

\begin{example}
In \ref{ex:h=2,k=5}, the maximal chains in $\PP_{2,5}=\PP_{2,5}[\hat 0,\hat 1]$
are listed in the lexicographic order, top down. 
\end{example}

\begin{remark}
\label{rem:Q*-via-tableaux}
\ Under the canonical isomorphism described in 
\ref{rem:chains-into-syt},
the maximal chains in $\PP_{h,k}[\ba,\bb]$ correspond to the standard Young tableaux
of a fixed skew shape $\lambda/\mu$, 
with $\lambda$ and~$\mu$ corresponding to $\bb$ and $\ba$, respectively. 
The lexicographic ordering on $\operatorname{MaxChains}(\PP_{h,k}[\ba,\bb])$
translates into the linear order on the standard tableaux of shape $\lambda/\mu$ defined as follows. 
Let $Q$ and $Q'$ be two such tableaux, 
and let $i$ be the smallest entry whose locations in $Q$ and~$Q'$
differ from each other. 
Specifically, let $b$ and $b'$ be the boxes containing $i$ in $Q$ and~$Q'$, respectively. 
Note that $b$ and~$b'$ are located in different rows and different columns. 
Then 
\begin{equation}
\label{eq:bb'}
Q'\prec Q \stackrel{\rm def}{\Longleftrightarrow}
\text{$b'$ is located to the left of~$b$.} 
\end{equation}

In the case $\PP=\PP_{h,k}[\ba,\bb]$ under our consideration,
the definition~\ref{eq:Q*}  of the chain~$Q^*$ 
translates into the language of tableaux as follows:
the elements of $Q^*$ are in bijection with the \emph{descents} of~$Q$,
i.e., those entries~$j$ for which $j\!+\!1$ appears in~$Q$ 
strictly to the left of~$j$---so~that switching $j$ and $j\!+\!1$ 
yields a lexicographically smaller tableau. 
More precisely, each descent~$j$ contributes a column $\cc\in Q^*$
corresponding to the Young diagram formed by the entries $1,\dots,j$ of~$Q$. 
See \ref{fig:P25}. 
\end{remark}

The following statement is (a reformulation of) a special case 
of the well known result on shellability of intervals 
in Bruhat order quotients for finite Coxeter groups, 
see~\citet{Bjorner-Wachs, Bjorner-Wachs-88, Proctor}. 
We provide a proof in order to keep the exposition self-contained. 

\begin{lemma}
\label{lem:shelling-Phk}
The lexicographic ordering on $\operatorname{MaxChains}(\PP_{h,k}[\ba,\bb])$
is proper. 
\end{lemma}

\begin{proof}
In the language of Young diagrams 
fitting inside the $h\times (k-h)$ rectangle (cf.\ \ref{rem:Q*-via-tableaux}), 
the claim \ref{eq:proper-ordering} translates into the following statement: 
if $Q$ and~$Q'$ are standard tableaux of skew shape~$\lambda/\mu$
such that $Q'\prec Q$, 
then there is a descent~$j$ in~$Q$ such that the entries $1,\dots,j$ in~$Q$ 
form a shape different from the one formed by those entries in~$Q'$. 
To prove this, consider the smallest entry~$i$ whose locations in $Q$ and~$Q'$ 
differ from each other.
Let $b$ and $b'$ denote the corresponding boxes, as in~\ref{eq:bb'}. 
By construction, the tableau~$Q$ must have a descent~$j\ge i$;
let us consider the smallest such descent.
None of the entries $1,\dots,j$ in~$Q$ is located in box~$b'$. 
On the other hand, in the tableau~$Q'$, the box $b'$ contains~$i\le j$,
and the claim follows. 
\end{proof}

Combining \ref{lem:gf-shelling} and \ref{lem:shelling-Phk}
enables us to express a generating function for multichains in $\PP_{h,k}[\ba,\bb]$
in terms of complete homogeneous symmetric functions.
These expressions, reformulated in terms of semistandard tableaux, 
will be used in \ref{sec:schur-via-multichains}
to obtain efficient $\{+,\times\}$-algorithms for computing Schur functions.


\section{Schur polynomials as multichain generating functions}
\label{sec:schur-via-multichains}

Let us recall the combinatorial definition of a Schur polynomial 
$s_\lambda(x_1,\dots,x_k)$ labeled by an integer partition
$\lambda=(\lambda_1\ge \cdots\ge \lambda_\ell\ge 0)$. 
Note that we allow trailing zeroes at the end of~$\lambda$. 

We assume that $\ell\le k$. 
This condition does not restrict the generality, since
$\lambda_\ell>0$ and $\ell>k$ imply 
$s_\lambda(x_1,\dots,x_k)=0$. 

We~use the notation 
$n=|\lambda|=\lambda_1+\cdots+\lambda_\ell$
for the size of the partition~$\lambda$.

\begin{definition}[\emph{Tableaux, Schur functions}]
\label{def:tableaux-schur}
\ A \emph{semistandard Young tableau} $T$ of shape $\lambda=|T|$ is an array of integers
\[
T=(t_{i,j} \mid 1\le i\le \ell,\ 1\le j \le \lambda_i)
\]
satisfying $t_{i,j} <
t_{i+1,j}$ and $t_{i,j} \le t_{i,j+1}$ whenever these inequalities
make sense. 
A~tableau $T$ is called \emph{standard} if each of the numbers
$1,\dots,n$ appears exactly once among the $n$ tableau
entries~$t_{i,j}\,$. 
We denote by $\xx^T$ the monomial associated with~$T$: 
\[
\xx^T=\prod_{i,j} x_{t_{i,j}}.
\] 
The \emph{Schur function} (or \emph{Schur polynomial}) 
$s_\lambda(x_1,\dots,x_k)$ is the generating function for semistandard  tableaux
of shape~$\lambda$ and entries in $\{1,\dots,k\}$: 
\[
s_\lambda(x_1,\dots,x_k)=\sum_{|T|=\lambda} \xx^T.
\]
By construction, $s_\lambda(x_1,\dots,x_k)$ is a homogeneous 
polynomial of degree~$n$ in the variables $x_1,\dots,x_k$, 
with positive integer coefficients. 
It is well known \citep[Chapter~7]{EC2} that $s_\lambda(x_1,\dots,x_k)$ is symmetric
with respect to permutations of the variables.
\end{definition}

\begin{example}
Let $\ell=2$ and $\lambda=(r,r)$. 
A semistandard tableau of shape $\lambda$ is a $2\times r$ matrix
$T=(t_{i,j})$ with positive integer entries which weakly increase
left-to-right in each row, and strictly increase top-down in each
column. 
The corresponding Schur polynomial is given by
$s_{(r,r)}(x_1,\dots,x_k)=\sum_T \prod_i \prod_j x_{t_{i,j}}$
where the sum is over all such tableaux with entries $\le k$. 
For example, if $r=2$ and $k=3$, then we get
$6$~different tableaux, and the answer is 
$s_{(2,2)}(x_1,x_2,x_3)=x_1^2x_2^2+x_1^2x_3^2+x_2^2x_3^2+x_1^2x_2x_3+x_1x_2^2x_3+x_1x_2x_3^2\,$.
\end{example}

Our next goal is to restate \ref{def:tableaux-schur} using the
language of multichain generating functions introduced in
\ref{sec:shelling}. 

The connection between Schur functions and the posets $\PP_{h,k}$ 
comes from the straightforward observation that the multichains of size~$m$ in $\PP_{h,k}$ are in 
a canonical bijection
with the semistandard tableaux of rectangular shape $h\times m$ and
entries~$\le k$. 
(This bijection should not be confused with the construction 
described in \ref{rem:chains-into-syt} above, which is of a rather different nature.) 
We next extend this correspondence to arbitrary shapes. 
This will require some preparation. 

\begin{definition}[\emph{Dissecting Young diagrams into rectangular shapes}]
\label{def:strike-columns}
Let $\lambda=(\lambda_1\ge \cdots\ge \lambda_\ell)$ be an integer partition.
As usual, we denote by~$\lambda'$ the conjugate partition,
i.e., the partition whose parts are the column lengths of (the shape of)~$\lambda$. 
We~then denote by $\tlambda'_1>\cdots>\tlambda'_s$ the integers, listed in the
decreasing order, which appear as parts of~$\lambda'$. 
In other words, $\tlambda_1',\dots,\tlambda_s'$ are all the
different heights of columns
in the Young diagram of~$\lambda$. 
We~denote by $\tlambda=(\tlambda_1\ge \cdots\ge \tlambda_\ell)$
the partition conjugate to $\tlambda'=(\tlambda_1',\dots,\tlambda_s')$. 
To rephrase, the shape~$\tlambda$ is obtained from~$\lambda$ by 
keeping one column of each height,
and striking out the rest. 

We can now dissect the Young diagram~$\lambda$ by vertical cuts into
$s$ rectangular shapes of sizes $h\times(\lambda_h-\lambda_{h+1})$ 
where $h$ runs over the set of parts of~$\tlambda'$
(equivalently, the distinct column lengths of~$\lambda$). 
To simplify notation for the sake of future arguments, 
we denote $h_j=\tlambda'_j$ and
$m_j=\lambda_{h_j}-\lambda_{h_j+1}-1$, 
so that $\lambda$ gets dissected into rectangles of sizes $h_j\times
(m_j+1)$, for $j=1,\dots,s$. 
\end{definition}

\begin{example}
\label{example:strike-columns}
Let $\lambda=(6,6,4,1,1)$, $\ell=5$.
Then 
\begin{equation*}
\lambda'=(5,3,3,3,2,2), \ \tlambda'=(5,3,2),\ \tlambda=(3,3,2,1,1),\ s=3. 
\end{equation*}
The shape $\lambda$ can be dissected by vertical cuts into three rectangles of
sizes $5\times 1$, $3\times 3$, and $2\times 2$, respectively. 
In this example, we have $h_1=5, h_2=3, h_3=2, m_1=0, m_2=2, m_3=1$. 
\end{example}

\begin{definition}[\emph{Pruning of tableaux}]
Let $T$ be a semistandard tableau of shape~$\lambda$.
The \emph{pruning} of~$T$ is the semistandard tableau~$\tT$ of shape~$\tlambda$
obtained from~$T$ 
by selecting the rightmost column of each height (and removing all
columns of that height located to the left of it). 
We denote by $\ba_1,\dots,\ba_s$ the columns of~$\tT$, listed left to
right. (These columns have heights $h_1,\dots,h_s$, respectively.)
We denote by $\bar\ba_j$ the column of height $h_{j+1}$ obtained
from~$\ba_j$ by removing the $h_j -h_{j+1}$ bottom entries. 

We furthermore denote by $T_1,\dots,T_s$ the semistandard tableaux of rectangular
shapes $h_1\times m_1,\dots,h_s\times m_s$ obtained by dissecting~$T$
by the vertical cuts described in \ref{def:strike-columns},
and then removing the rightmost column from each of the resulting
tableaux. (If $m_j=0$, then $T_j$ is empty.) 
Thus $T$ is obtained by interlacing the rectangular tableaux~$T_j$
with the columns of the pruning:
$T=[T_1|\ba_1|T_2|\ba_2|\cdots|T_s|\ba_s]$.   
\end{definition}

\begin{example}
Continuing with \ref{example:strike-columns}, 
let $T=\smallbmatrix{
1 & 1 & 2 & 2 & 2 & 4\\ 
2 & 2 & 3 & 3 & 3 & 5\\
4 & 5 & 6 & 6\\
5\\
6}$.
Then $\tT\!=\!\smallbmatrix{
1 & 2 & 4\\
2 & 3 & 5\\
4 & 6\\
5\\
6
}$, 
$T_1\!=\!\varnothing$,
$\ba_1\!=\!\smallbmatrix{
1 \\ 
2 \\
4 \\
5\\
6}$, $T_2\!=\!\smallbmatrix{
1 & 2 \\ 
2 & 3 \\
5 & 6 
}$,
$\ba_2\!=\!\smallbmatrix{
2 \\ 
3 \\
6
}$,
$T_3\!=\!\smallbmatrix{
2 \\ 
3 
}$,
$\ba_3\!=\!\smallbmatrix{
4\\ 
5
}$. 
\end{example}

\pagebreak[3]

Consider the set of semistandard tableaux~$T$ of a given
shape~$\lambda$, with entries~$\le k$, and 
with a given pruning~$\tT=[\ba_1|\cdots|\ba_s]$. 
Note that once $\tT$ and~$\lambda$ have been fixed, 
each tableau~$T_j$, for $1\le j\le s$, can be chosen
independently of the others, 
as long as it satisfies the following restrictions:
\begin{itemize}
\item
$T_j$ is a semistandard tableau 
of rectangular shape~$h_j\times m_j$, with entries~$\le k$;
as such, it can be viewed as a multichain of size~$m_j$ in the
poset~$\PP_{h_j,k}$;
\item
every column~$\ba$ in~$T_j$
(i.e., every element of this multichain) satisfies the inequalities  
$\bar\ba_{j-1}\le\ba\le\ba_j$, with respect to the partial
order in~$\PP_{h_j,k}$.
\end{itemize}
(We set $\bar\ba_0=\hat 0=\smallbmatrix{1\\[-.05in] \vdots\\ \ell}$ by convention, 
so that the lower bound is redundant for~\hbox{$\!j=\!1$}.) 
This gives a bijection between the set of tableaux
under consideration and the Cartesian product of sets of multichains in
the posets $\PP_{h_j,k}$: 
\[
\left\{
\begin{array}{c}
\text{semistandard tableaux $T$}\\
\text{of shape~$\lambda$, with entries~$\le k$,}\\
\text{with pruning $\tT=[\ba_1|\cdots|\ba_s]$}
\end{array}
\right\}
\longleftrightarrow
\prod_{j=1}^s
\left\{
\begin{array}{c}
\text{multichains}\\ 
\text{of size $m_j$}\\
\text{in $\PP_{h_j,k}[\bar\ba_{j-1},\ba_j]$}
\end{array}
\right\}
\]
Identifying multichains in $\PP_{h_j,k}[\bar\ba_{j-1},\ba_j]$
with semistandard tableaux of rectangular shape, 
and passing to generating functions, we obtain the following result.

\begin{lemma}
\label{lem:schur-via-pruning}
With the notation as above, we have 
\begin{equation}
\label{eq:schur-via-pruning}
s_\lambda(x_1,\dots,x_k)
=\sum_{\tT} \xx^\tT 
     \prod_{j=1}^s \sum_{T_j} \xx^{T_j}, 
\end{equation}
where
\begin{itemize}
\item
$\tT=[\ba_1|\cdots|\ba_s]$ runs over semistandard tableaux 
of shape~$\tlambda$, with~entries~\hbox{$\le k$};
\item 
each $T_j$ runs over semistandard tableaux of rectangular shape $h_j\times m_j$ 
whose col\-umns form a multichain in $\PP_{h_j,k}[\bar\ba_{j-1},\ba_j]$. 
\end{itemize}
\end{lemma}

In view of \ref{lem:gf-shelling}
and \ref{lem:shelling-Phk}, 
the sums $\sum_{T_j} \xx^{T_j}$ appearing in 
\ref{eq:schur-via-pruning} can be computed 
using the formula \ref{eq:multichains-via-h}:  

\begin{lemma}
\label{lem:sandwiched-tableaux}
Let $\ba,\bb\in\PP_{h,k}$ be two columns such that $\ba\le\bb$. 
Then
\begin{equation}
\label{eq:sandwiched-tableaux}
\sum_T \xx^T
=\sum_Q \xx^{Q^*} h_{m-|Q^*|}(\xx^{\cc_1},\dots,\xx^{\cc_N}),
\end{equation}
where
\begin{itemize}
\item 
$T$ runs over semistandard tableaux of rectangular shape $h\times m$ 
whose columns form a multichain in $\PP_{h,k}[\ba,\bb]$;  
\item
$Q=[\cc_1|\cdots|\cc_N]$ runs over the maximal chains in $\PP_{h,k}[\ba,\bb]$;
\item
$Q^*$ is given by~\ref{eq:Q*}.
\end{itemize}
\end{lemma}

For the reader's convenience, we restate the definition of~$Q^*$ in concrete terms;
cf.\ also \ref{rem:Q*-via-tableaux}. 
For each pair of consecutive columns $\cc_j$ and~$\cc_{j+1}$, 
we have $\cc_{j+1}=\cc_j+\mathbf{e}_{i_j}$ for some $i_j\in\{1,\dots,h\}$, 
where $\mathbf{e}_i$ denotes the column whose $i$th component is equal to~$1$,
and all others are equal to~$0$. 
The chain/tableau $Q^*$ is formed by the subset of columns~$\cc_j$
for which $i_{j-1}>i_j$ and moreover $\cc_{j-1}+\mathbf{e}_{i_j}\in\PP_{h,k}$
(so that replacing $\cc_j$ by $\cc_{j-1}+\mathbf{e}_{i_j}$ transforms~$Q$
into a lexicographically smaller maximal chain). 

\begin{example}
\label{ex:h=2,k=5,continued}
Let $h=2$, $k=5$, $\ba=\hat 0$, $\bb=\hat 1$, cf.\ \ref{ex:h=2,k=5}.
Then \ref{eq:sandwiched-tableaux} becomes 
\begin{align*}
s_{(m,m)}(x_1,\!\dots\!,x_5)\!=
h_{m}(x_1x_2,x_1x_3,x_2x_3,x_2x_4,x_3x_4,x_3x_5,x_4x_5)\\
+x_2x_5\, h_{m-1}(x_1x_2,x_1x_3,x_2x_3,x_2x_4,x_2x_5,x_3x_5,x_4x_5) \\
+x_1x_4\,h_{m-1}(x_1x_2,x_1x_3,x_1x_4,x_2x_4,x_3x_4,x_3x_5,x_4x_5) \\
+x_1x_4\!\cdot \!x_2x_5 \, h_{m-2}(x_1x_2,x_1x_3,x_1x_4,x_2x_4,x_2x_5,x_3x_5,x_4x_5)\\
+ x_1x_5 \, h_{m-1}(x_1x_2,x_1x_3,x_1x_4,x_1x_5,x_2x_5,x_3x_5,x_4x_5). 
\qed
\end{align*}
\end{example}


\section{Proof of the main theorem} 
\label{sec:proof-main}

Combining \ref{eq:schur-via-pruning}
and~\ref{eq:sandwiched-tableaux}, we obtain:

\begin{corollary}
The Schur polynomial $s_\lambda(x_1,\dots,x_k)$ is given by 
\begin{equation}
\label{eq:schur-via-h}
s_\lambda(x_1,\dots,x_k)
=\sum_{|\tT|=\tlambda} \xx^\tT 
     \prod_{j=1}^s \sum_Q \xx^{Q^*} h_{m_j-|Q^*|}(\xx^{\cc_1},\dots,\xx^{\cc_N}), 
\end{equation}
where
\begin{itemize}
\item
$\tlambda$, $s$, $h_1,\dots,h_s$, and $m_1,\dots,m_s$ 
are described in \ref{def:strike-columns}; 
\item
$\tT=[\ba_1|\cdots|\ba_s]$ runs over semistandard tableaux 
of shape~$\tlambda$, with~entries~\hbox{$\le k$};
\item
$Q=[\cc_1|\cdots|\cc_N]$ runs over the maximal chains in $\PP_{h_j,k}[\bar\ba_{j-1},\ba_j]$.
\end{itemize}
\end{corollary}

To prove \ref{th:main}, we analyze the (semiring) complexity 
of computing a Schur polynomial $s_\lambda(x_1,\dots,x_k)$ 
using the formula~\ref{eq:schur-via-h} together with \ref{th:h_n}. 

We begin by computing the monomials $\xx^\cc$, for all columns~$\cc$ 
of height~$h_j$ with entries~$\le k$, for each $j\le s$. 
This can be done using $\le\ell \sum_{j\le s} \binom{k}{h_j}$ multiplications. 
(Note that $s\le\ell$.) 

Recall that the Young diagram $\tlambda$ has $s$ columns, of heights $h_1,\dots,h_s$. 
Hence the number of tableaux $\tT$ appearing in~\ref{eq:schur-via-h}
is bounded by $\prod_{j\le s} \binom{k}{h_j}$. 

Each monomial $\xx^\tT$ can be computed by $s-1$ multiplications
(given all the $\xx^{\cc_i}$). 

The number of maximal chains in $\PP_{h_j,k}[\bar\ba_{j-1},\ba_j]$ is at most $h_j^{h_j(k-h_j)}$, by
\ref{lem:bound-max-chains}.
Each of these chains has length $N\le h(k-h)+1$.
Since $|Q^*|\le|Q|=N$, we can compute $\xx^{Q^*}$ in time $\le h(k-h)$. 
Also, $m_j-|Q^*|\le\lambda_1$. 
\ref{th:h_n} now implies that we can compute
$\xx^{Q^*} h_{m_j-|Q^*|}(\xx^{\cc_1},\dots,\xx^{\cc_N})$
in time $O(h^2(k-h)^2\log(\lambda_1))$. 
Putting everything together, we obtain
the following upper bound on the semiring complexity of $s_\lambda(x_1,\dots,x_k)$:
\[
\ell \sum_{j\le s} \binom{k}{h_j}
+ \prod_{j\le s} \binom{k}{h_j}
\cdot(2s
+\sum_{j\le s} h_j^{h_j(k-h_j)} (O(h_j^2(k-h_j)^2\log(\lambda_1))) ).  
\]
This can be replaced by 
$O(\log(\lambda_1)) s\ell^2 k^2 2^{ks}\ell^d$
where 
\[
d=\max_j h_j(k-h_j)=\max_j \lambda'_j(k-\lambda'_j),
\]
and then by $O(\log(\lambda_1)) k^5 2^{k\ell}\ell^d$. 
\qed

\begin{acknowledge}
Partially supported by the NSF grant DMS-1361789 (S.~F.), 
the RSF grant 16-11-10075 (D.~G.), and the NSERC (\'E.~S.). 
D.~G.\ thanks MCCME Moscow and Max-Planck Institut f\"ur Mathematik 
for their hospitality and inspiring atmosphere.

We thank the referee for a number of suggestions 
which led to the improvement of the presentation. 
\end{acknowledge}

\end{document}